\documentclass[conference]{IEEEtran}
\usepackage{amsmath}
\usepackage{mathrsfs}
\usepackage{pifont}
\usepackage{bbding}
\usepackage{amsmath,epsfig}

\usepackage{stmaryrd}
\usepackage{amssymb}
\usepackage{amsfonts}
\usepackage{epic}
\usepackage{graphicx}
\usepackage{curves}
\usepackage{cite}

\usepackage{algorithm}
\usepackage{algpseudocode}

\usepackage{stfloats}
\usepackage{latexsym}
\usepackage{epstopdf}
\usepackage{epic}
\usepackage{bm}
\usepackage{multirow}
\usepackage{xcolor}
\begin{document}
\newcommand{\reffig}[1]{Fig. \ref{#1}}
\newcommand{\figref}[1]{\figurename~\ref{#1}}

\newtheorem{definition}{Definition}
\newtheorem{theorem}{Theorem}
\newtheorem{corollary}{Corollary}
\newtheorem{proposition}{Proposition}
\newtheorem{lemma}{Lemma}
\newtheorem{remark}{Remark}
\renewcommand{\algorithmicrequire}{\textbf{Input:}}  
\renewcommand{\algorithmicensure}{\textbf{Output:}}  
\renewcommand{\thefootnote}{}


\title{Joint Tx/Rx  Energy-Efficient Scheduling  in Multi-Radio  Networks: A Divide-and-Conquer Approach }

\author{\IEEEauthorblockN{Qingqing Wu, Meixia Tao, and Wen Chen}
\IEEEauthorblockA{Department of Electronic Engineering, Shanghai Jiao Tong University, Shanghai, China.\\ Emails: \{wu.qq,mxtao,wenchen\}@sjtu.edu.cn.}}

\maketitle
\begin{abstract}
Most of the existing works on energy-efficient wireless communication systems only consider  the transmitter (Tx) or the receiver (Rx) side power consumption but not both. Moreover, they often assume the static circuit power consumption. To be more practical, this paper considers the joint Tx and Rx power consumption  in multiple-access radio networks, where the power model takes both the transmission power and the dynamic circuit power into account. We  formulate the joint Tx and  Rx energy efficiency (EE) maximization problem which is
 a combinatorial-type  one due to the indicator function for scheduling users and activating radio links. The link EE and the user EE are then introduced  which have the similar structure as the system EE. Their  hierarchical  relationships are exploited to tackle the problem using a divide-and-conquer approach, which is only of  linear complexity.
We further  reveal  that the static receiving power plays a critical  role in the user scheduling. 
  Finally, comprehensive numerical  results are provided to validate the theoretical findings and demonstrate the  effectiveness of the proposed algorithm for improving the system EE. \footnote {This work is supported by the National 973 Project \#2012CB316106, by NSF China  \#61322102,  \#61161130529, and \#61328101, by the STCSM Science and Technology Innovation Program \#13510711200, by the SEU National Key Lab on Mobile Communications \#2013D11.  Wen Chen is also with the School of Electronic Engineering and Automation, Guilin University of Electronic Technology.}
\end{abstract}

\section{Introduction}
The increasing number of new wireless access devices and various  services lead to a significant increase in the demand for  higher user data rate. While the higher energy consumption is a great concern as well for future wireless communication systems. Recently, there has been an upsurge of interest in the energy efficiency (EE) optimization field.
Basic concepts of energy-efficient communications are  introduced in \cite{wu2012green} and several advanced physical layer techniques for EE are studied in \cite{xiao2013qos,miao2013energy1,ng2012energy2,cheung2012achieving,zhang2011energy}.

However, all the above works only consider one side power consumption, i.e., either the transmitter (Tx) or the receiver (Rx) side.
In fact, the expectation of  limiting  electric expenditure and reducing carbon emissions requires the base station  to perform in an energy-efficient manner \cite{wu2012green}, while  minimizing the user side energy consumption also deserves more efforts due to  capacity limited batteries and user experience requirements \cite{kim2010leveraging,luo2013joint}. Moreover,
according to \cite{auer2010d2},  the techniques adopted to improve the EE of one end of the communication system may adversely affect the EE of the other end. Therefore, it is necessary to consider the joint Tx and Rx EE optimization,  which shall provide  more flexibility for the energy saving at the side interested or both.

For EE  oriented research, one of the most important tasks is to quantify the power consumption of the communication system \cite{auer2010d2}.
 Most of the existing works only consider a constant circuit power so as to simplify the system analysis and make the problem more tractable.
 However, it has been reported in \cite{wu2012green,gruber2009earth} that a rough modeling for the power consumption   can not  reflect the true behaviour of
wireless devices and thus might provide  misleading conclusions. Therefore, the power consumption modeling should not only capture the key
system components but also characterize the reality \cite{wu2012green}.

The main contributions of this paper are summarized as follows:
  1) We formulate the joint Tx and Rx EE maximization problem in which  the link dependent signal processing power, the static circuit power as well as the
transmission power  are considered based on a comprehensive study  \cite{auer2010d2}, while the power model of existing works
    \cite{xiao2013qos,miao2013energy1,cheung2012achieving,ng2012energy2
,xu2013energy,mao2013energy} are basically special cases.
2)  We explore the fractional structure of the system EE and introduce the concept of the  individual EE, i.e., the link EE and the user EE. Based on these,  an optimal approach of  linear complexity is proposed to solve the non-convex EE maximization problem. Moreover, this approach can also be used to optimally solve the problem in \cite{miao2013energy1} where only a quadratic complexity method is proposed.
3) We reveal that the static receiving power has an implicit interpretation of the optimal number of scheduled users.  In the extreme case when the static receiving power is negligible, time division multiplexing access (TDMA) is optimal for the energy-efficient transmission.
\section{System Description  and Problem Formulation}
\subsection{System Model}
 Consider a multi-user multi-radio  network, where  $K$ users  are communicating with one access point (AP) over $M$ orthogonal radio links simultaneously.
  It is assumed that each user $k$, for $k=1,...,K$, is assigned prior with a fixed subset of radio links, denoted as $\mathcal{M}_k$,  and  that the radio links of different users do not overlap with each other so as to void interference, i.e., $\mathcal{M}_k\bigcap \mathcal{M}_m ={\O}$. The multiple radio links can be formed by orthogonal multiplexing techniques, such as frequency division multiplexing.
The channel between the AP and each user is assumed to be quasi-static fading and they are all  equipped with one antenna.
It is assumed that the perfect and global channel state information (CSI) of all users is available for the AP, which allows us to do energy-efficient scheduling.
The channel gain of user $k$ over  link $i$ and the corresponding power allocation of this link are  denoted as $g_{k,i}$ and $p_{k,i}$, respectively. The receiver noise is modelled as a circularly symmetric complex Gaussian random variable with zero mean and variance $\sigma^2$ for all links. Then the data rate of user $k$ over  link $i$, denoted as $r_{k,i}$, can be expressed as
\begin{equation}\label{eq1}
r_{k,i}=B\log_{2}\left(1+\frac{p_{k,i}g_{k,i}}{\Gamma\sigma^2}\right),
\end{equation}
where $B$ is the bandwidth of each radio link and $\Gamma$ characterizes the gap between  actual achievable rate and channel capacity due to a practical modulation and coding design \cite{miao2013energy1}. Consequently, the overall system data rate can be expressed as
\begin{equation}\label{eq2}
R_{\rm{tot}}=\sum^{K}_{k=1}\omega_k R_k=\sum^{K}_{k=1}\omega_k\sum_{i\in \mathcal{M}_k} r_{k,i},
\end{equation}
where $R_k$ is  the data rate of user $k$ and $\omega_k$ which is provided by upper layers, represents the  priority of user $k$.

\subsection{Joint Tx/Rx Power Consumption Model }
In this work, we adopt the power consumption model from \cite{auer2010d2} published by Energy Aware Radio and neTwork tecHnologies (EARTH)  project, which provides a comprehensive  characterization of the power consumption for each component involved in the communication.

At the user side, the power dissipation consists of two parts, i.e., the transmission power and the circuit power. Denote $P_{Tk}$ as the overall transmission power of user $k$ and it is given by
\begin{equation}\label{eq3}
P_{Tk}=\frac{\sum_{i\in \mathcal{M}_k} p_{k,i}}{\xi},
\end{equation}
 where $\xi\in(0,1]$ is a constant which accounts for the efficiency of the power  amplifier.
Denote $P_{Ck}$ as the circuit power of user $k$. According to \cite{auer2010d2}, the circuit power of each device contains a dynamic part for the signal processing which linearly scales with the number of active links, and a static part independent of links for other circuit blocks, i.e.,
\begin{equation}\label{eq4}
P_{Ck}(n_k^o)=n_k^oP_{{\rm{dyn}}, k}+\mathcal{I}(n_k^o)P_{{\rm{sta}}, k},
\end{equation}
where $n_k^o$ is the number of active  links and can be expressed as
$n_k^o=\sum_{i\in \mathcal{M}_k} \mathcal{I}(p_{k,i})$. Here,
the indicator function $\mathcal{I} (x)$ is defined as
\begin{equation}\label{eq5}
\mathcal{I}(x) =\left\{
\begin{array}{lcl}
1,& \text{if}~ x>0,\\
0,& \text{otherwise}.
\end{array}\right.
\end{equation}
 Specifically, if $p_{k,i}>0$, then $ \mathcal{I}(p_{k,i})=1$ means that  link $i$ is active, and if $n_k^o>0$, then $ \mathcal{I}(n_k^o)=1$ means that user $k$ is scheduled.
 In (\ref{eq4}), $P_{{\rm{dyn}}, k}$  and $P_{{\rm{sta}}, k}$ are dynamic and static components of the  circuit power for user $k$, respectively.
  Considering different types of terminals  in practical systems,  $P_{{\rm{dyn}}, k}$  and $P_{{\rm{sta}},k}$ can be different  for different user $k$. Now, the overall power consumption of user $k$, denoted as $P_k$, is
\begin{equation}\label{eq6}
P_k=P_{Tk}+P_{Ck}(n_k^o).
\end{equation}

 At the AP side,  the receiving circuit power consumption also consists of two similar parts as the user device \cite{kim2010leveraging,auer2010d2}. Denote $P_{\rm{dyn, 0}}$ and $P_{\rm{sta,0}}$ as the  dynamic  and static receiving circuit power, respectively.  Then the overall power consumption at the AP side can be expressed as
\begin{equation}\label{eq7}
P_{0}=\sum_{k=1}^{K}n^o_kP_{{\rm{dyn}},0}+P_{{\rm{sta}},0}.
\end{equation}

Finally, the overall  power consumption of the  system can be expressed as
\begin{align}\label{eq8}
P_{\rm{tot}} =\sum^{K}_{k=1} P_k+P_0.
\end{align}
\subsection{Problem Formulation}
Energy efficiency is commonly defined by the ratio of the overall system rate $R_{\rm{tot}}$ over the overall system power consumption $P_{\rm{tot}}$ \cite{xiao2013qos,miao2013energy1,cheung2012achieving}.
  Our goal  is to jointly optimize the  user scheduling, the link activation and the power control to maximize the  EE of the considered system.
Mathematically, we can formulate the EE optimization problem  as (P1)
\begin{align}\label{eq10}
\mathop {\max }\limits_{\bm{p}} &~~\frac{\sum_{k=1}^{K} \omega_k\sum_{i\in \mathcal{M}_k} B\log_{2}\left(1+\frac{p_{k,i}g_{k,i}}{\Gamma\sigma^2}\right)}
{\sum^{K}_{k=1} \left( \frac{\sum_{i\in \mathcal{M}_k} p_{k,i}}{\xi}+ P_{Ck}(n_k^o) + n_k^oP_{{\rm{dyn}},0}\right)+ P_{{\rm{sta}},0}} \nonumber  \\
\text{s.t.} \,&~~ n_k^o=\sum_{i\in \mathcal{M}_k} \mathcal{I}(p_{k,i}),  ~~~ 1\leq k \leq K,  i\in \mathcal{M}_k,   \nonumber \\
& ~~0\leq p_{k,i}\leq P^{k,i}_{\mathop{\max}}, ~~~ 1\leq k \leq K,  i\in \mathcal{M}_k, 
\end{align}
where $\bm{p} \triangleq$ \{$p_{k,i}|k=1,2,...K;i\in \mathcal{M}_k$\}. For practical consideration, we assume that each radio link $i$ of user $k$ has a maximum allowed transmit power $P^{k,i}_{\mathop{\max}}$. Note that the authors  in \cite{miao2013energy1} consider a similar problem formulation  to (\ref{eq10}) but without power constraint which is thereby  a special case of this paper.

The existence of the two layered indicator functions, i.e., $\mathcal{I}(p_{k,i})$ and $\mathcal{I}(n_k^o)=\mathcal{I}\left(\sum_{i\in \mathcal{M}_k} \mathcal{I}(p_{k,i})\right)$ makes the objective function discontinuous and hence non-differentiable. The global optimal solution of (\ref{eq10}) is generally difficult to be obtained with an efficient complexity. In the following section, we explore the particular structure of system EE and  show that the global optimal solution can actually be obtained using a divide-and-conquer approach with low complexity.

\section{Energy-Efficient  Scheduling}
In this section, we  solve the system EE maximization problem directly from a fractional-form  perspective. This idea  results from the  connection of the   EE from three levels, namely,  \emph{the link  Energy Efficiency}, \emph{the user Energy Efficiency}, and \emph{the system Energy Efficiency}.
\subsection{Link Energy Efficiency and User Energy Efficiency}

 \emph{Definition 1}  (Link Energy Efficiency): The EE of link $i$ of  user $k$,  for  $i\in \mathcal{M}_k$, $k=1, ..,K$,  is defined as the ratio of  the weighted achievable rate of the user on this link over the consumed power associated with this link, i.e.,
\begin{eqnarray}\label{eq11}
ee_{k,i}=\frac{\omega_kB\log_{2}\left(1+\frac{p_{k,i}g_{k,i}}{\Gamma\sigma^2}\right)}{\frac{p_{k,i}}{\xi}+ P_{{\rm{dyn}},k}+ P_{{\rm{dyn}},0}},
\end{eqnarray}
where  the link-level power consumption counts the transmission power of the user over the link, per-link dynamic circuit power of the user and the AP, respectively.

 It is easy to prove that this fractional type function have the stationary point which is also the optimal point \cite{Boyd}. By setting the derivative of $ee_{k,i}$ with respect to $p_{k,i}$ to zero, we obtain that the optimal power value ${p}_{k,i}$ and the optimal link EE under peak power constraint satisfies
\begin{align}\label{eq13}
p^*_{k,i}=\left[\frac{ B\xi \omega_k}{ ee^*_{k,i}\ln2}-\frac{\Gamma\sigma^2}{g_{k,i}}\right]^{P^{k,i}_{\mathop{\max}}}_0, \forall\, k,   i\in \mathcal{M}_k,
\end{align}
where  $[x]^a_b\triangleq \min\left\{\max\{x,b\},a\right\}$.
Note that $\frac{ B\xi \omega_k}{ ee^*_{k,i}\ln2}>\frac{\Gamma\sigma^2}{g_{k,i}}$, i.e.,
$p^*_{k,i}>0$ always holds for $ee^*_{k,i}$, since otherwise $ee^*_{k,i}$ would be zero.
Based on (\ref{eq11}) and (\ref{eq13}), the numerical values of  ${ee}^*_{k,i}$ and $p_{k,i}^*$  can be easily obtained by the bisection method.

 \emph{Definition 2}  (User Energy Efficiency): The EE of user $k$, for $k=1,...,K$, is defined as the ratio of the weighted total achievable rate of the user on all its preassigned radio links over the total power consumption associated with this user, i.e.,
\begin{align}
EE_{k}=\frac{\omega_k\sum_{i\in\mathcal{M}_k} B\log_{2}\left(1+\frac{p_{k,i}g_{k,i}}{\Gamma\sigma^2}\right)}{ \frac{\sum_{i\in\mathcal{M}_k}  p_{k,i}}{\xi}+n_k^o( P_{{\rm{dyn}}, k}+ P_{{\rm{dyn}},0})+P_{{\rm{sta}}, k}},
\end{align}
where the user-level power consumption counts the total transmission power of the user, the overall circuit power of the user and the dynamic processing power of the AP related to this user.

Now, we find the optimal power control to maximize the user EE. The problem is formulated as
\begin{align}\label{eq14}
\mathop {\max }\limits_{\{p_{k,i}\}} ~~~ &~~~~~~ EE_k  \nonumber \\
\text{s.t.}~~~~~ & n_k^o=\sum_{i\in \mathcal{M}_k} \mathcal{I}(p_{k,i}),  ~~ 1\leq k \leq K,  i\in \mathcal{M}_k,   \nonumber \\
~~~& p_{k,i}\leq P^{k,i}_{\mathop{\max}}, ~~  1\leq k \leq K,  i\in \mathcal{M}_k,   \nonumber \\
~~~& p_{k,i} \geq 0,  ~~~~ 1\leq k \leq K,  i\in \mathcal{M}_k.
\end{align}
Define $\Phi_k$ as the set of active links  for user $k$ and then $n^o_k$ is the cardinality of $\Phi_k$. Given any $\Phi_k$, it is easy to prove that $EE_{k}$ is strictly quasiconcave in ${p_{k,i}}$. Thus, similar to the link EE, the optimal power allocation under set $\Phi_k$  satisfies
\begin{align}\label{eq15}
p_{k,i}=\left[\frac{ B\xi \omega_k}{EE^*_k\ln2}-\frac{\Gamma\sigma^2}{g_{k,i}}\right]^{P^{k,i}_{\mathop{\max}}}_0, \forall\, k,  i\in \Phi_k.
\end{align}
Note that if $p_{k,i}=0$, it suggests that this link should not be active in the optimal solution, but its corresponding circuit power $ P_{{\rm{dyn}}, k}+ P_{{\rm{dyn}},0}$ has already been accounted in calculating the total power consumption in (\ref{eq14}). Therefore, we have to obtain the set $\Phi_k$ in which all radio links are allocated with strictly positive powers  in maximizing $EE_k$.

Let $EE^*_{\Phi_k}$ denote the optimal intermediate user EE of user $k$ when its current set of active links is $\Phi_k$, and then the value of $EE^*_{\Phi_k}$ can be obtained by (\ref{eq14}) and (\ref{eq15}). The next theorem provides a general condition for determining whether an arbitrary link should be scheduled.
\begin{theorem}\label{user}
 For any link $i\notin \Phi_k$, if $ EE^*_{\Phi_k}\leq ee_{k,i}^{*} $, then there must be $EE^*_{\Phi_k}\leq EE^*_{\Phi_k \bigcup \{(k,i)\}}\leq ee_{k,i}^{*}$, and  the  link $i$ should be activated and added to  $\Phi_k$; else if $ EE^*_{\Phi_k}> ee_{k,i}^{*}$, then there must be $EE^*_{\Phi_k}> EE^*_{\Phi_k \bigcup \{(k,i)\}}> ee_{k,i}^{*}$, and  the  link $i$ should not be activated and added to  $\Phi_k$.
\end{theorem}
\begin{proof}
  Please see Appendix \ref{apdix1}.
\end{proof}

The interpretation is also obvious: the new link $i$ should have a better utilization of the power than  its user.
In what follows, we introduce how to obtain the optimal user EE based on the link EE, and the details of this procedure are summarized in line 1-14 of  Algorithm 1.

Sort all  radio links of user $k$ according to their link EE $ee^*_{k,i}$ in descending order, i.e., $ee^*_{k,1} \geq ee^*_{k,2}\geq... \geq ee^*_{k,n_k}$, and set the initial $\Phi_k={\O}$. Then  we successively take one link from the order and judge whether it should be added to $\Phi_k$. Until some link is determined not  to be activated or all links are activated, then based on the current $\Phi_k$, we can obtain the optimal user EE.

\begin{remark}
The optimality of the proposed procedure for maximizing the user EE is ensured by the ordering of the link EE as well as the conclusion of Theorem \ref{user}.
This idea opens up a new way to address the fractional-form EE maximization problem.
\end{remark}
\subsection{User Scheduling and Link Adaptation}
In this subsection, we show how to solve the original problem (\ref{eq10}) based on the link EE and the user EE.
 For the explanation convenience, we first introduce two auxiliary sets. Denote  $\Phi$  as the set of active links of all  users, i.e., $\Phi=\{(k,i)\,|\,p_{k,i}>0, \forall i, k \}$, with its optimum denoted as $\Phi^*$. Denote $U$ as the set of scheduled users  which have at least  one active link belonging to set $\Phi$, i.e., $U=\{\,k\,|\,(k,i)\in \Phi, \forall k, i\}$. Apparently, $U$ can be sufficiently determined by $\Phi$.

Given the set of overall active links $\Phi$, and accordingly the set of scheduled users $U$, then $n_k^o$ can be readily calculated  and
problem (\ref{eq10})
is simplified into the following problem
\begin{align}\label{eq141}
\mathop {\max }\limits_{\bm{p}} &~~\frac{\sum_{k\in U} \omega_k\sum_{i \in \Phi} B\log_{2}\left(1+\frac{p_{k,i}g_{k,i}}{\Gamma\sigma^2}\right)}
{\sum_{k\in U} \left( \frac{\sum_{i\in \Phi} p_{k,i}}{\xi}+P_{Ck}(n_k^o)+ n_k^oP_{{\rm{dyn}},0}\right)+  P_{{\rm{sta}},0}}  \nonumber \\
\text{s.t.}\,&~~ 0< p_{k,i}\leq P^{k,i}_{\mathop{\max}}, ~~~ 1\leq k \leq K,  i\in \Phi.
\end{align}
\begin{algorithm}[!hbpt]
 \caption{Energy-Efficient Scheduling Algorithm}
  \begin{algorithmic}[1]
    \item
 {\bf{for}}  $k=1:K$\\
          ~~~Compute  $ee^*_{k,i}$   for all $i\in \mathcal{M}_k$, by (\ref{eq11})  and (\ref{eq13});  \\
          ~~~~Sort all  links of user $k$ in  descending order of  $ee^*_{k,i}$,\textcolor{white}{xxxx}  \textcolor{white}{xx}i.e., $ee^*_{k,1} \geq ee^*_{k,2}\geq... \geq ee^*_{k,n_k}$;  \\ \nonumber
          ~~~Set  $\Phi_k={\O}$ and $EE_{\Phi_k}^{*}=0$; \\
          ~~~{\bf{for}}  $i=1:n_k$\\
         ~~ ~\ \ \ \ \ {\bf{if}}   \  \  $EE_{\Phi_k}^{*}\leqslant ee^*_{k,i}$ \ \ {\bf{do}}    \\
         ~~  ~\ \ \ \ \ \ \ \quad  $\Phi_k={\Phi_k\bigcup \{(k,i)\}}$ ;           \\
          ~~   ~\ \ \ \ \ \  \ \quad Compute $p^*_{k,i}$  and $EE^*_{\Phi_k}$ by (\ref{eq14}) and (\ref{eq15}); \\
         ~~     ~\ \ \ \ \ {\bf{else}}       $EE_{\Phi_k}^{*}> ee^*_{k,i}$    \\
         ~~     ~\ \ \ \ \ \ \ \quad   ${\Phi_k}^{*} = \Phi_k$;\\
         ~~     ~\ \ \ \ \ \ \ \quad   $EE_{\Phi^*_k}^{*} = EE_{\Phi_k}^{*}$; \ \ {\bf{return}} \\
         ~~     ~\ \ \ \ \ {\bf{end} }\\
          ~~~{\bf{end} }\\
                    {\bf{end} }\\
                    Sort all users (include both real users and virtual users )  in descending order of $EE_{\Phi^*_k}^{*}$, i.e., $EE_{\Phi^*_1}^{*} \geq EE_{\Phi^*_2}^{*}\geq,...,\geq EE_{\Phi^*_L}^{*}$; \\
Set  $\Phi={\O}$, $U={\O}$, and $EE_{\Phi}^{*}=0$; \\
{\bf{for}}  $k=1:L$\\
          ~~~~~{\bf{if}}   \  \  $EE_{\Phi}^{*}\leqslant EE_{\Phi^*_k}^{*}$ \ \ {\bf{do}}    \\
         ~~     ~\ \ \ \ \ $\Phi={\Phi\bigcup \Phi^*_k}$  and $U=U\bigcup \{k\}$ ;           \\
         ~~     ~\ \ \ \ \ Obtain $p^*_{k,i}$  and $EE^*_{\Phi}$ by solving problem (\ref{eq141}); \\
          ~~~~~{\bf{else}}       $EE_{\Phi}^{*}> EE_{\Phi^*_k}^{*}$   \\
         ~~     ~\ \ \ \ \ $\Phi^{*}=\Phi$;\\
         ~~     ~\ \ \ \ \ $EE_{\Phi^{*}}^{*}= \ EE_{\Phi}^{*}$;   \ \ {\bf{return}} \\
          ~~~~~{\bf{end} }\\
{\bf{end} }
\end{algorithmic}\label{alg1}
\end{algorithm}
 Obviously, problem (\ref{eq141}) can be verified  as a standard quasiconcave optimization problem and thereby can be readily solved as (\ref{eq14}). Then our task is transformed to find the scheduled users and its corresponding active links.
Recall that in obtaining the optimal user EE, some links may not  be activated and for all inactive links, use $(k', i')$ to denote them. Then we define each inactive link, say link $i'$ of user $k'$ as a virtual user $\ell$ just like the real users in the system, and  let $\{(k',i')\}=\Phi^*_{\ell}$.  Therefore, the EE of this virtual user $\ell$ is exactly the  EE of  link $i$ of user $k$, i.e.,
$EE^*_{\Phi_{\ell}}=ee^*_{k,i}$.  In the rest, unless specified otherwise,  term ``user" refers to both real users and virtual users. The difference between the real user and the virtual user is that each real user may contain several links and its circuit power includes the static user scheduling power  $P_{{\rm{sta}}, k}$ as well as the link-dependent power
 $ P_{{\rm{dyn}}, k}+P_{{\rm{dyn}},0}$,
while each virtual user only contain one link and its circuit power thereby is given by
 $ P_{{\rm{dyn}},k}+P_{{\rm{dyn}},0}$.

We first sort all users in descending order according to the user EE $EE_{\Phi^*_k}^{*}$, i.e., $EE_{\Phi^*_1}^{*} \geq EE_{\Phi^*_2}^{*}\geq,...,\geq EE_{\Phi^*_L}^{*}$, where $L$ is the overall number of real users and virtual users. Then, we have the following lemma to characterize a property of the order.
\begin{lemma}
Assume that the virtual user $\ell$ is derived from the link $i'$ of the real user $k$. Following the descending order of the user EE,
the order index  of this virtual user $\ell$   must be  larger than that of its associated  real user $k$.
\end{lemma}
\begin{proof}
According to the user EE, we have $EE^*_{\Phi^*_k}> ee_{k,i'}^{*}$, i.e, $ EE^*_{\Phi^*_k}> EE^*_{\Phi_{\ell}}$.  Therefore, when  they are mixed together to specify the order, the virtual user $\ell$ (inactive link) must be ranked after its  corresponding real user $k$.
\end{proof}

This lemma guarantees that those virtual users (inactive links) of real user $k$ must be less likely to be active in the system EE compared with the user $k$ (the active links in the user EE), otherwise it may lead to the case that some link is scheduled finally in the system, but its associated real user is not scheduled, which contradicts the reality.

 In the following, we explore the special structures of the system EE, the user EE and the link EE,  and show how to  obtain the optimal set $\Phi^*$.
In each round, we add one user to the set $U$  following the order and  add all its active links in $\Phi^*_k$ to $\Phi$, respectively.
Then  based on $\Phi$,   the optimal system EE $EE_{\Phi}^{*}$ can be calculated as  (\ref{eq141}). By the following theorem, we obtain the maximum system  EE of problem (\ref{eq10}).
\begin{theorem}\label{system}
 1) For any $\Phi_k^*\notin \Phi$, If $ EE^*_{\Phi}\leq EE_{\Phi^*_k}^{*}$, then there must be $EE^*_{\Phi}\leq EE^*_{\Phi \bigcup \Phi^*_k}\leq EE_{\Phi^*_k}^{*}$;
 else if  $ EE^*_{\Phi}> EE_{\Phi^*_k}^{*}$, then there must be $ EE^*_{\Phi}> EE^*_{\Phi\bigcup \Phi^*_k}>  EE_{\Phi^*_k}^{*}$;
2) If any user $k$ is scheduled, then all active links in terms of the optimal user EE will also be activated in maximizing system EE.
\end{theorem}
\begin{proof}
  Please see Appendix \ref{apdix2}.
\end{proof}

The first statement  suggests that in each round, the comparison result of $EE^*_{\Phi}$ and  $EE_{\Phi^*_k}^{*}$ is necessary and sufficient to determine whether the $k$th user can be scheduled to improve the system EE. While the second statement guarantees the optimality of the active links in maximizing the system  EE.
This theorem guarantees the optimality of the proposed method which exhibits the concept of divide-and-conquer following the EE of three levels. The process of method is summarized in  Algorithm 1 and it is easy to show that the complexity of the divide-and conquer approach overall has a linear complexity of the power control.
\subsection{Impact of Static Receiving  Power on User Scheduling}
The next theorem reveals the  relationship between the  user scheduling and the static receiving power.
\begin{theorem}\label{pr0}
1) The optimal number of scheduled users in maximizing the system EE is nondecreasing with the static receiving power $P_{{\rm{sta}},0}$;
2) When $P_{{\rm{sta}},0}$ is  negligible, i.e., $P_{{\rm{sta}},0}\rightarrow0$,  TDMA is optimal for energy-efficient transmission;
3) When $P_{{\rm{sta}},0}$ is sufficiently large, all users will be scheduled for energy-efficient transmission.
\end{theorem}
\begin{proof}
Due to the space limitation, we only provide a sketch of the proof here.  It is easy to show that the system EE is decreasing with the  static receiving power $P_{{\rm{sta}},0}$. Then, from Theorem \ref{system}, we can show that less users would be scheduled for a higher system EE.  A  more detailed  proof  will be given  in the journal version of this paper.
\end{proof}

The intuition is that  when $P_{{\rm{sta}},0}$ is larger, the additional power consumption brought from scheduling users is less dominant, which makes it more effective to achieve higher EE.
If there is no additional power consumption for operating systems,  i.e., $P_{{\rm{sta}},0}=0$,  the optimal energy-efficient strategy is only to schedule the ``best'' user where the best is in terms of the user EE. It has the similar interpretation as that of the throughput maximization problem in TDMA systems: only the user with the best channel gain will be scheduled.
\begin{table}[!t]
\centering
\caption{\label{table1}  SYSTEM PARAMETERS} \label{tb1}
\renewcommand\arraystretch{1}
\begin{tabular}{|c|c|}
\hline
{Parameter} & {Description} \\
\hline
Carrier frequency &      $2$ GHz     \\
\hline
Bandwidth of each radio  link, $B$&     $ 15$ kHz      \\
\hline
Maximal allowed transmit power, $P^{k,i}_{\mathop{\max}}$&     $ 25$ dBm      \\
\hline
Static circuit  power of the AP, $P_{{\rm{sta}},0}$  &      $5000$ mW\\
\hline
Link dependent power of the AP, $P_{{\rm{dyn}},0}$&      $45$ mW\\
\hline
Static circuit power of  user $k$, $P_{{\rm{sta}}, k}$ &    $100$ mW \\
\hline
Link dependent power of  user $k$, $P_{{\rm{dyn}}, k}$&      $5\-- 30$ mW\\
\hline
Power density of thermal noise variance  &    $-174$ dBm/Hz \\
\hline
Power amplifier efficiency, $\xi$ &    $0.38$    \\
\hline
Cell radius, $r$&     $ 1000$ m      \\
\hline
Path loss model  &    Okumura-Hata      \\
\hline
Penetration loss  &    $20$ dB    \\
\hline
Lognormal shadowing  &    $8$ dB    \\
\hline
Fading  &    Rayleigh flat fading  \\
\hline
\end{tabular}
\end{table}
 \begin{figure}[!th]
\centering
\includegraphics[width=3.5in ]{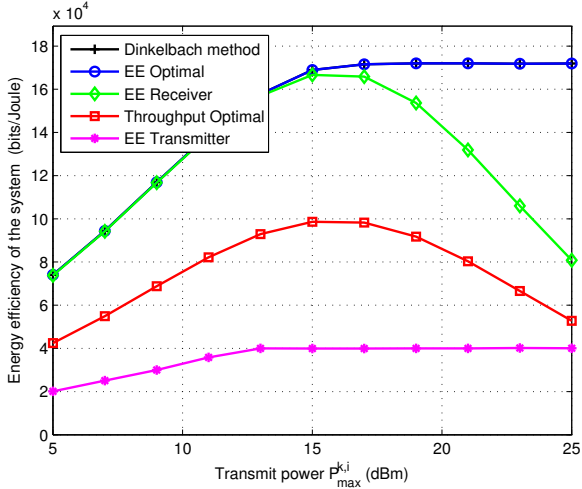}
\caption{The system EE  versus the transmit power.}\label{transmit_power}
\end{figure}
From Theorem \ref{pr0}, it is also interesting to note that  the number of users scheduled can not be  guaranteed although  the weights have been imposed on users, especially for the case with low static receiving power.
\section{Numerical Results }

 \begin{figure}[!th]
\centering
\includegraphics[width=3.5in ]{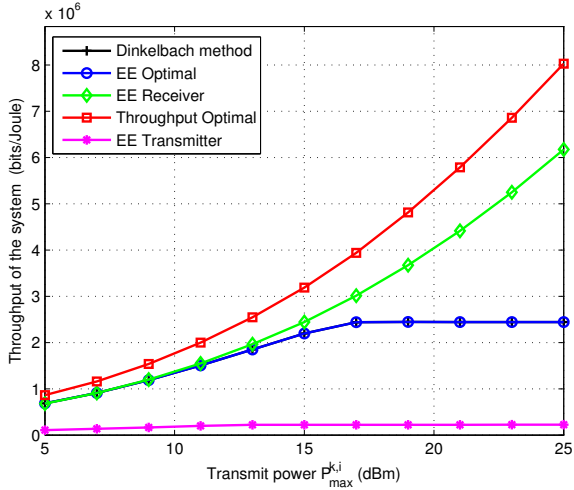}
\caption{The system throughput  versus the transmit power.}\label{throughput}
\end{figure}
\begin{figure}[!th]
\centering
\includegraphics[width=3.5in ]{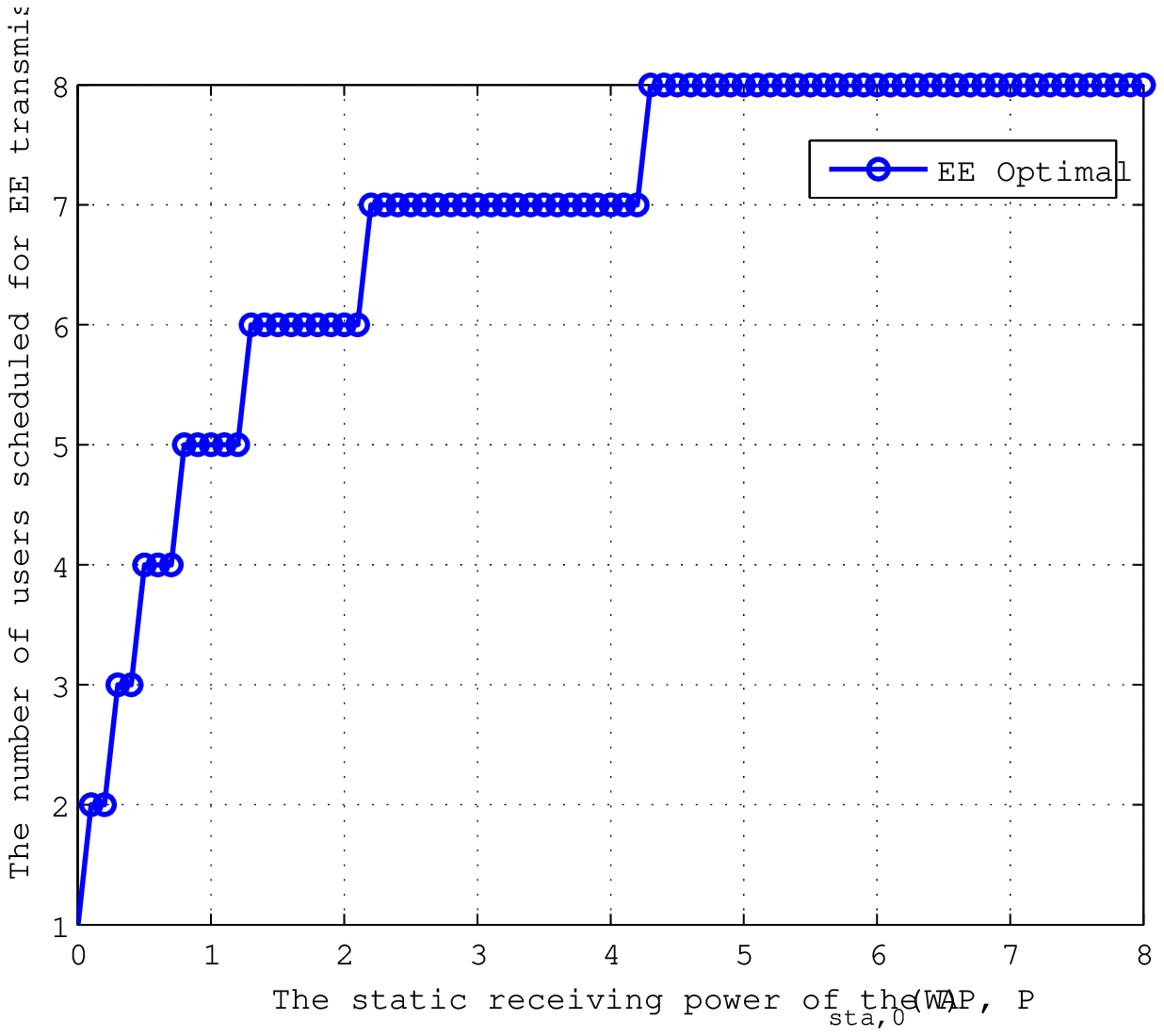}
\caption{The number of scheduled users  versus the static receiving power.}\label{user_number}
\end{figure}

In this section, we provide simulation results to validate our theoretical findings and  demonstrate the effectiveness  proposed  methods. There are eight equally weighted users in the system and each user is configured with twenty radio links. Without loss of generality,  $P^{k,i}_{\mathop{\max}}$ is assumed the same for all users and $\Gamma$ is assumed as 1. Other system parameters are listed in Table I  according to \cite{zhang2010joint,xu2013energy}  unless specified otherwise.

In \figref{transmit_power}, we compare the EE of the following methods:  1) Dinkelbach method: the existing optimal method \cite{ng2012energy2}; 2) EE Optimal: joint Tx and Rx optimization; 3) EE  Transmitter: based on the Tx side optimization \cite{qing1}; 4) EE  Receiver: based on the Rx side optimization  \cite{qing1}; 5) Throughput Optimal: based on the throughput maximization.  In \figref{transmit_power}, we can first observe that our proposed method performs the same as the Dinkelbach method, which demonstrates its optimality. Moreover, as the transmit power increases,  the performance of the EE Optimal scheme first increases and then approaches a  constant because of its energy-efficient nature, while those of the Throughput Optimal scheme and the EE Receiver scheme first increase and then decrease due to their greedy use of power.
 It is also interesting to note that the EE Receiver scheme approaches the EE Optimal scheme in the  low transmit power regime while it is more close to the Throughput Optimal scheme in the high transmit power regime.
A similar phenomenon can also be found in \figref{throughput} in terms of the system throughput. Moreover, the EE Transmitter scheme results in both low EE and spectral efficiency due to the fact that  \emph{only one user is scheduled}, which has been theoretically shown in Section III-C.


 \figref{user_number} further demonstrates our theoretical findings in Theorem \ref{pr0} which
 characterizes the monotonicity of the number of users scheduled  with $P_{{\rm{sta}},0}$.  We observe that when $P_{{\rm{sta}},0}$ is negligible,  the optimal energy-efficient strategy is  to schedule only one user.  As $P_{{\rm{sta}},0}$ increases, more users are scheduled to improve the system EE through boosting the system throughput.



\section{Conclusions}
This paper  investigated the joint transmitter and receiver EE maximization problem in multi-radio networks. A holistic and dynamic power consumption model  was established for the considered system. Then, the EE maximization problem is directly addressed from the fractional perspective, which results in an linear divide-and-conquer approach. Moreover, we pointed out  that the static receiving power has an implicit interpretation for the optimal  number of scheduled users. In the extreme case when the static receiving power is negligible, TDMA is the optimal scheduling strategy. In order to meet the QoS in practice, we then extended the propose method to solve the problem  with  minimal user  data rate constraints, which exhibits good performance with linear complexity.



\appendices

\section{Proof of Theorem \ref{user}}\label{apdix1}
Denote ${p}_{k,i}^*$ as the optimal power  corresponding to $ee^*_{k,i}$ by  (\ref{eq11})  and (\ref{eq13}).
Also, denote $\hat{p}_{k,i}$ and $\check{p}_{k,i}$ as the optimal powers  corresponding to $EE^*_{\Phi_k\bigcup \{(k,i)\}}$ and $EE^*_{\Phi_k}$ by (\ref{eq141}), respectively. Let $S_k \triangleq \{p_{k,i}|0 \leq p_{k,i} \leq P^{k,i}_{\mathop{\max}}, \forall \,  i\in \mathcal{M}_k, k=1,...,K\}$ and  $P_{k,i}(p_{k,i})=\frac{\Theta_tp_{k,i}}{\xi}+\Theta_{\rm{t}} P_{{\rm{dyn}},k}$.
Then, we have the following
\begin{eqnarray}\label{eq39}
&&EE^*_{\Phi_k\bigcup \{(k,i)\}}   \nonumber \\
&=& \max_{\bm{p_k}\in S_k}\frac{\sum_{\ell=1}^i \omega_kr_{k,\ell}(p_{k,\ell})}{\sum_{\ell=1}^{i} P_{k,i}(p_{k,\ell})+P_{{\rm{sta}}, k}}   \nonumber  \\
&=& \frac{\sum_{\ell=1}^{i-1} \omega_kr_{k,\ell}(\hat{p}_{k,\ell})+\omega_kr_{k,i}(\hat{p}_{k,i})}{\sum_{\ell=1}^{i-1} P_{k,\ell}(\hat{p}_{k,\ell})+P_{{\rm{sta}}, k}+P_{k,i}(\hat{p}_{k,i})}    \nonumber  \\
&\geq& \frac{\sum_{\ell=1}^{i-1}  \omega_kr_{k,\ell}(\check{p}_{k,\ell})+ \omega_kr_{k,i}(p_{k,i}^*)}{\sum_{\ell=1}^{i-1} P_{k,\ell}(\check{p}_{k,\ell})+P_{{\rm{sta}}, k}+P_{k,i}({p}_{k,i}^*)}   \nonumber  \\
&\geq& \min\left\{\frac{\sum_{\ell=1}^{i-1}  \omega_kr_{k,\ell}(\check{p}_{k,\ell})}{\sum_{\ell=1}^{i-1} P_{k,\ell}(\check{p}_{k,\ell})+P_{{\rm{sta}}, k}},\frac{ \omega_kr_{k,i}({p}_{k,i}^*)}{P_{k,i}({p}_{k,i}^*)} \right\}    \nonumber \\
&=& \min\left\{EE^*_{\Phi_k}, ee^*_{k,i}\right\}.
\end{eqnarray}
 On the other hand,
 \begin{eqnarray}\label{eq391}
&&EE^*_{\Phi_k\bigcup \{(k,i)\}}     \nonumber  \\
&=& \frac{\sum_{\ell=1}^{i-1}  \omega_kr_{k,\ell}(\hat{p}_{k,\ell})+ \omega_kr_{k,i}(\hat{p}_{k,i})}{\sum_{\ell=1}^{i-1} P_{k,\ell}(\hat{p}_{k,\ell})+P_{{\rm{sta}}, k}+P_{k,i}(\hat{p}_{k,i})}   \nonumber  \\
&\leq& \max\left\{\frac{\sum_{\ell=1}^{i-1}  \omega_kr_{k,\ell}(\hat{p}_{k,\ell})}{\sum_{\ell=1}^{i-1} P_{k,\ell}(\hat{p}_{k,\ell})+P_{{\rm{sta}}, k}},\frac{ \omega_kr_{k,i}(\hat{p}_{k,i})}{P_{k,i}(\hat{p}_{k,i})}\right\}    \nonumber \\
&\leq& \max\left\{\frac{\sum_{\ell=1}^{i-1}  \omega_kr_{k,\ell}(\check{p}_{k,\ell})}{\sum_{\ell=1}^{i-1} P_{k,\ell}(\check{p}_{k,\ell})+P_{{\rm{sta}}, k}},\frac{ \omega_kr_{k,i}({p}_{k,i}^*)}{P_{k,i}({p}_{k,i}^*)}\right\}    \nonumber \\
&=& \max\left\{EE^*_{\Phi_k},  ee^*_{k,i}\right\}.
\end{eqnarray}
Based on (\ref{eq39}) and (\ref{eq391}), we have
\begin{equation}\label{keyneq}
 \min\left\{EE^*_{\Phi_k}, ee^*_{k,i}\right\} \leqslant EE^*_{\Phi_k\bigcup \{(k,i)\}}\leqslant \max\left\{EE^*_{\Phi_k},  ee^*_{k,i}\right\}.
\end{equation}
By (\ref{keyneq}), Theorem \ref{user} can be easily proved.
\section{Proof of Theorem \ref{system}}\label{apdix2}
The statement 1) in Theorem \ref{system} can be similarly proved by an extension of Theorem \ref{user}, thus we omit them for brevity.
We now prove 2) by contradiction. Assume that user  $k$ is scheduled,  but the link $i$, for $i\in \Phi_k^*$,  is not activated in maximizing the system EE, i.e., $(k,i)\notin \Phi^*$. In Theorem \ref{user}, we have shown that the sufficient and necessary condition of any $i\in \Phi_k^*$ is that $EE^*_{\Phi_k}\leq ee_{k,i}^{*}$. Thus, it follows that
\begin{align}\label{ap_eq1}
EE^*_{\Phi^*_k}\leq ee^*_{k,n^o_k} \leq ee^*_{k,i},  ~~ \forall i\in \Phi^*_k,
\end{align}
where $n^o_k$ also denotes the last link activated according to the link EE order since user $k$ overall has $n^o_k$ links activated.
On the other hand, since user $k$ is scheduled, we must have $EE^*_{\Phi}\leq EE_{\Phi^*_k}$  by Theorem \ref{system}. Combining with (\ref{ap_eq1}), it follows that
\begin{align}
EE^*_{\Phi^*}\leq EE^*_{\Phi^*_k} \leq ee^*_{k,n^o_k} \leq ee^*_{k,i}=EE^*_{\Phi_{\ell}},  ~~ \forall i\in \Phi^*_k,
\end{align}
where the virtual user expression is adopted, i.e.,$\{(k,i)\}=\Phi^*_{\ell}$.
According to 1) in Theorem \ref{system},  there must be
\begin{align}\label{t_eq2}
EE^*_{\Phi^*}\leq EE^*_{\Phi^* \bigcup \Phi^*_{\ell}}=EE^*_{\Phi^* \bigcup \{(k,i)\}}.
\end{align}
Therefore, from (\ref{t_eq2}),  we can conclude that scheduling  the link $i$ of user $k$ should be scheduled in maximizing the system EE, which contradicts the assumption that  $(k,i)\notin \Phi^*$.
\bibliographystyle{IEEEtran}
\bibliography{IEEEabrv,mybib}

\end{document}